\documentclass[conference]{IEEEtran}
\usepackage{url}
\usepackage[dvips]{graphicx}
\usepackage{amssymb,amsfonts,amsmath,amsthm,amscd,dsfont,mathrsfs}
\usepackage{graphicx,float,psfrag,epsfig,color}

\newcommand\independent{\protect\mathpalette{\protect\independent}{\perp}} 
\def\independent#1#2{\mathrel{\rlap{$#1#2$}\mkern2mu{#1#2}}}

\newcommand{\iid}{\stackrel{\text{iid}}{\sim}}

\newcommand{\ft}{\mathcal{F}}
\newcommand{\m}{\mathrm{M}}
\newcommand{\spa}{\mathrm{Spa}}
\newcommand{\sm}{\widehat{\mathrm{M}}}
\newcommand{\dou}{\mathrm{Doub}}
\newcommand{\cir}{\mathrm{Circ}}

\newcommand{\F}{\mathbb{F}} 
\newcommand{\mR}{\mathbb{R}} 

\newcommand{\mZ}{\mathbb{Z}}
\newcommand{\mC}{\mathbb{C}}

\newcommand{\pp}{\mathbb{P}}
\newcommand{\E}{\mathbb{E}}

\newcommand{\e}{\varepsilon}

\DeclareMathOperator{\diag}{diag}
\newcommand{\cS}{\mathcal{S}}
\newcommand{\cT}{\mathcal{T}}

\theoremstyle{definition}
\newtheorem{definition}{Definition}
\theoremstyle{plain}
\newtheorem{thm}{Theorem}
\theoremstyle{plain}
\newtheorem{prop}{Proposition}
\theoremstyle{plain}
\newtheorem{lemma}{Lemma}
\theoremstyle{plain}
\newtheorem{corol}{Corollary}
\theoremstyle{plain}
\theoremstyle{remark}
\newtheorem{remark}{Remark}
\theoremstyle{discussion}
\theoremstyle{plain}

\begin{document}

\title{Universal Polar Coding and Sparse Recovery}

\author{
  \IEEEauthorblockN{Emmanuel Abbe}
  \IEEEauthorblockA{Information Processing Group, EPFL\\
    Email: emmanuel.abbe@epfl.ch}
}


\maketitle

\begin{abstract}
  \boldmath
This paper investigates universal polar coding schemes. In particular, a notion of ordering (called convolutional path) is introduced between probability distributions to determine when a polar compression (or communication) scheme designed for one distribution can also succeed for another one. The original polar decoding algorithm is also generalized to an algorithm allowing to learn information about the source distribution using the idea of checkers. These tools are used to construct a universal compression algorithm for binary sources, operating at the lowest achievable rate (entropy), with low complexity and with guaranteed small error probability.  

In a second part of the paper, the problem of sketching high dimensional discrete signals which are sparse is approached via the polarization technique. It is shown that the number of measurements required for perfect recovery is competitive with the $O(k \log (n/k))$ bound (with optimal constant for binary signals), meanwhile affording a deterministic low complexity measurement matrix. 

\end{abstract}


\section{Introduction}

A new technique called `polarization' has recently been introduced by Ar{\i}kan in \cite{ari} to construct efficient channel coding schemes. The codes resulting from this technique, called polar codes, have several nice attributes: (1) they are linear codes generated by a low-complexity deterministic matrix (2) they can be analyzed mathematically and bounds on the error probability (exponential in the square root of the block length) can be {\it proved} (3) they have a low encoding and decoding complexity (4) they allow to reach the Shannon capacity on any discrete memoryless channels (DMC). These codes are indeed the first codes with low decoding complexity that are provably capacity achieving on any DMC.

Since \cite{ari}, the polarization technique has been generalized to various settings. For example, it has been used in \cite{korada} for rate-distortion via duality with test channels, in \cite{elgamal,shamai,vardy} for wiretap channels and information secrecy, and in \cite{2mac,mmac} for a multi-user problem (multiple accessing). 

In this paper, we investigate the problem of robustness of the polar coding schemes with respect to the knowledge of source or channel distribution. The perfect knowledge of this distribution is never available, and it is important that any potentially practical scheme shows some robustness to this knowledge. We hence develop several tools to construct universal polarization schemes. 

We then consider the problem of sketching high-dimensional sparse signal using the polarization technique. The hope being to leverage properties (1)-(4) to construct a deterministic low-complexity sketching matrix and an efficient sparse recovery algorithm. Since the method is defined for signals valued in finite sets, it is of interest to lift the construction to the real setting. Yet in this paper, we focus our attention on the sketching problem for signals that are discrete, motivated by applications dealing with such signals. 
This occurs for example in network monitoring problems \cite{varghese,gilbert}. 
We will see that, just like one can exploit sparsity in the domain, the sparsity in the magnitude (signals taking values in finite sets) can be exploited to develop an efficient sketching method via the polarization technique. 

Some results in this paper have been presented in \cite{abbe}.   
\subsection{Channel and source polarization}

Ar{\i}kan shows in \cite{ari} that an arbitrary binary input discrete memoryless channel $W$ can be polarized as follows: $n$ independent uses of $W$ can be transformed into $n$ successive uses of synthesized binary input channels that have (except for a vanishing fraction) a symmetric capacity 
which tends to either 0 or 1 (with $n$). 
In \cite{qpol}, this result is generalized to $q$-ary input alphabets where $q$ is prime, and in \cite{mmac} it is extended to $q$ being powers of two (considering $q$ to be a power of two has computational advantages, but the case of powers of prime follows too). We state here the result of \cite{qpol} for $q$ prime. Notation: $X^n :=(X_1,\dots,X_n)$.
\begin{thm}\label{ari}
Let $W$ be a $q$-ary input discrete memoryless channel with $q$ prime, $n$ a power of 2, and let $U^n$ be i.i.d. uniform random variables on $\F_q$. Let $X^n =  U^n G_n$, where $G_n=    \bigl[\begin{smallmatrix} 
      1 & 0 \\
      1 & 1 \\
   \end{smallmatrix}\bigr]^{\otimes \log_2(n)}$,
and $Y^n$ be the output of $n$ independent uses of $W$ when the input is $X^n$. 
Then, for any $\delta \in (0,1)$,
\begin{align}
&\frac{1}{n}|\{i: I(U_i; Y^n U^{i-1}) > \delta\}| \stackrel{n \to \infty}{\longrightarrow} I(W),
\end{align}
where $I(W)$ is the mutual information of $W$ (with a uniform input distribution).
\end{thm}
Theorem \ref{ari} can then be used to show the following polarization phenomenon for sources. 
\begin{thm}\label{ari2}
Let $X^n$ be $n$ i.i.d. random variables with distribution $p$ on $\F_q$, $n$ a power of 2, and let $U^n=X^n G_n$, where $G_n=    \bigl[\begin{smallmatrix} 
      1 & 0 \\
      1 & 1 \\
   \end{smallmatrix}\bigr]^{\otimes \log_2(n)}$. Then, for any $\delta \in (0,1)$, 
\begin{align}
&\frac{1}{n} |\{i: H(U_i | U^{i-1}) > \delta\}|  \stackrel{n \to \infty}{\longrightarrow}  H(p),
\end{align}
where $H(p)$ is the entropy of the distribution $p$. 
\end{thm}
We will see in Section \ref{duality} that previous result follows from Theorem \ref{ari} via a duality argument. A slightly more general result is presented in \cite{ari3}.

Note that all entropies and mutual informations are computed with logarithms in base $q$ (where $q$ is the input or source alphabet size).

{\it A coding scheme from Theorem \ref{ari}.} The limit in the theorem implies that for $n$ large enough and except for a vanishing fraction of indices $i$, $I(U_i; Y^n U^{i-1})$ must be close to either 0 or 1. Hence, this suggests a coding scheme: on the indices $i$ for which the channel is good, i.e., $I(U_i; Y^n U^{i-1})$ is close to 1, put uncoded information bits in $U_i$, and for the other indices, put frozen bits that are predetermined and revealed to the decoder. This defines the vector $U^n$. Then, the vector $X^n$ is sent over $n$ independent uses of the channel. Note that the rate of this code is given by the logarithm of the number of information bits in $U^n$ divided by $n$, and by Theorem \ref{ari}, this can be made arbitrarily close to $I(W)$. Now the receiver knows two things: the location\footnote{No analytical formula is known to compute these indices. They are found with algorithms, as in \cite{vardy2}} of the indices $i$ containing information and frozen bits, and the value of the frozen bits (on symmetric channels, the frozen bits can be all chosen to be zero).
Hence, from the output $Y^n$ of $X^n$, the receiver starts by decoding the first component $U_i$ which is not frozen. By virtue of Theorem \ref{ari}, one of the two possible value of $U_i$ will have (w.h.p.) a probability close to one and hence, the decoder has a small probability of decoding $U_i$ incorrectly. This process is then iterated to decode successively the entire vector $U^n$. An analysis of the scaling\footnote{To show achievability, the speed of convergence to the polarized channels matters, and it is shown to be roughly $2^{-\sqrt{n}}$ in \cite{ari2}.} of the bit error probability (decoding wrongly a component in $U_i$) allows to conclude that w.h.p. errors cannot propagate in this scheme, and hence, this scheme achieves the uniform mutual information of the channel. 
A remarkable feature of this coding scheme is that the encoding and decoding complexity is shown to be $O(n \log n)$. 


{\it A coding scheme from Theorem \ref{ari2}.} The limit in the theorem implies that for $n$ large enough and except for a vanishing fraction of indices $i$, $H(U_i|U^{i-1})$ is close to either 0 or 1. Hence, the transformation $G_n$ extracts the randomness in $X^n$, which is initially uniformly dissipated over the $n$ components, into specific components indexed by the $i$'s such that $H(U_i|U^{i-1})$ is close to 1. 
 Lossless compression can then be performed as follows: from a given source output $X^n$, compute $U^n$ and store the components of $U^n$ which do not have an entropy close to 0. Note that, from Theorem \ref{ari2}, the compression rate can be made arbitrarily close to $H(p)$ (lowest possible rate). 
 For the reconstruction, since the components with entropy close to 0 can be recovered correctly with high probability given the past components, we can proceed successively in an analogue manner as for the channel decoding problem. The speed of polarization is shown to scale similarly as in the channel case, and again, the encoding and decoding complexity of this source coding scheme is only $O(n \log n)$. 


\subsection{Goals}

In this paper, we are interested in analyzing how sensitive the performance of the previous source/channel coding scheme is to the knowledge of the source/channel distribution. 
The knowledge of source/channel distribution is used at two moments for each problem. In the channel coding problem, it is first used to identify the location of the ``good channels'', or equivalently, the location of the indices $i$ where the information bits shall be sent. It is then used again in the decoding process, to compute the probabilities that an information bit $U_i$ is equal each element of $\F_q$ (from the polarization phenomenon, we know that one of these probabilities should have a probability close to 1, but one still needs to compute which one it is). Similarly, for source coding, the knowledge of the source distribution is first used to find the components of $U^n$ which must be stored, and then in the reconstruction part, to compute the value of each non-stored components. 

We will hence address the problem of constructing polar coding schemes which can compress losslessly sources without requiring perfect knowledge of their distributions, or which can communicate reliably over channels without requiring perfect knowledge of the channel distribution.
The application to the channel setting follows then from Section \ref{duality}, where the duality between the source and channel problem is made explicit. 

We will then consider the problem of sparse data recovery, using polar codes. From the discussion on the source polarization theorem above, a connection to the sketching problem is apparent: if we sense the signal $U^n$ only in the components $i$ for which $H(U_i|U^{i-1})$ is close to 1, we obtain a sampling of the signal which allows perfect recovery of the full signal, with a significantly reduced number of measurements. There are however several differences between a compressed sensing setting \cite{candes,donoho} and the source polarization setting; in particular, in the latter setting the source is random with a known distribution and it is valued in a finite field (of arbitrarily large cardinality), whereas it is real and with no prior (besides sparse) in compressed sensing. Hence, a first question is to ask how sparsity, i.e., the property of having many components that are 0, is modeled for  such random signals, and how much the choice of a specific sparse probability distribution matters. This part can be investigated using our results on universal source polarization, which establishes the connection between the two parts of this paper.    

\section{Results}
A universal compression algorithm for binary sources is introduced in Section \ref{main1}.
Theorem \ref{mainthm} shows that this algorithm performs at the lowest achievable rate (entropy), with a $O(n \log_2 n)$ complexity and (roughly) a $O(2^{-\sqrt{n}})$ error probability.

Partial generalizations are discussed for non-binary sources in Section \ref{main2}.

In Section \ref{univprior}, a low-complexity deterministic sketching matrix is constructed.
It is shown in Theorem \ref{csprop} that for $k$-sparse signals in $\F_a^n$, $O(k \log_e n/k)$ measurements taken with the proposed sketching matrix are sufficient to recover perfectly the original vector with a probability at least $1-O(2^{-\sqrt{n}})$, and a reconstruction algorithm of complexity $O(a \log_2 a \cdot n \log_2 n)$. An improved version of this Theorem (regarding the dependence in $a$ of the constants) is investigated in Section \ref{imp}.

\section{Duality between source and channel polarization}\label{duality}   
In this section, we connect Theorem \ref{ari} and Theorem \ref{ari2}.
Let $p$ be a distribution on $\F_q$ and consider using Theorem \ref{ari} for an additive noise channel, i.e., $Y=X\oplus Z$ for $Z$ distributed under $p$ and independent of $X$. We then have $Y^n = G_n U^n \oplus Z^n$ and
\begin{align}
&I(U_i; Y^n U^{i-1}) = 1 -H(U_i | Y^n U^{i-1}) \notag \\
& = 1-H((G_nY^n \ominus G_n Z^n)_i | Y^n (G_nY^n \ominus G_n Z^n)^{i-1}) \notag \\
& = 1-H((G_n Z^n)_i |  (G_n Z^n)^{i-1}). \label{drop}
\end{align}
Equality \eqref{drop} uses the fact that $Y^n$ is independent of $Z^n$ because $U^n$, and hence $G_nU^n$, are uniformly distributed over $\F_q$. We also use the fact that $G_n^{-1} = G_n$.
Hence, Theorem \ref{ari} and \eqref{drop} imply Theorem \ref{ari2}.

Stated as such, Theorem \ref{ari2} does not imply Theorem \ref{ari}, since additive noise channels are not representative of all possible channels. 
In \cite{ari3} a slightly more general result than Theorem \ref{ari2} is stated, where an auxiliary random variable $Y$ (side-information), which is a random variable correlated with $X$ but not intended to be compressed, is introduced in the conditioning of each entropy term. This could be used for the reverse implication. 

   

In this paper, we focus mostly on the source setting, since it is the ``simplest'' setting, hence the one to start with.
Using previous expansions, the results obtained in the source setting will directly admit a counter-part in the channel setting, for the case of additive noise channels.

\section{Defining orderings and mathematical preliminaries}
\begin{definition}[Measures]
Let $a$ be a prime integer, $\F_a:=\{0,1,\ldots, a-1\}$ and $\m(a)$ be the set of probability measures on $\F_a$.
For any $k \in \F_a$, let 
\begin{align*}
&\sm_k(a) := \{ p \in \m(a) : p(i)=p(j), \, \forall i,j \neq k, p(k) \geq \frac{a-1}{a} \}
\end{align*}
and $\sm(a) : = \cup_{k \in \F_a} \sm_k(a)$. We refer to the measure in $\sm(a)$ as the the spike measures.
\end{definition}

\begin{definition}[Matrices]
We denote by $\dou(a)$ the set of doubly stochastic matrices of size $a \times a$,
and by $\cir(a)$ the set of circulant stochastic matrices of size $a \times a$.
\end{definition}

\begin{definition}[Orders]
We define
\begin{align}
& p_1 \prec_h p_2 \quad \text{if} \quad h(p_1) \geq h(p_2), \\
& p_1 \prec_d p_2 \quad \text{if} \quad p_1 = D p_2 \text{ for } D \in \dou(a), \\
& p_1 \prec_c p_2 \quad \text{if} \quad p_1 = C p_2 \text{ for } C \in \cir(a).
\end{align}
Note that $\prec_d$ is the majorization order and $p_1 \prec_c p_2$ is equivalent to $p_1 = c \star p_2 $ for $c \in \m(a)$, where $\star$ denotes the circular convolution on $\F_a$. 
\end{definition}
Note that we use the term ``order'' in a broad sense here (not a mathematical order).
\begin{lemma}[Orders hierarchy]\label{hier} 
\begin{align} 
p_1 \prec_c p_2 \quad \Rightarrow \quad  p_1 \prec_d p_2 \quad \Rightarrow \quad  p_1 \prec_h p_2.
\end{align}
\end{lemma}
\begin{proof}
The first implication follows from the fact that $\cir(a) \subset \dou(a)$ and the second implication follows from the Schur-concavity of the entropy \cite{marshall}. 
\end{proof}
One can easily find examples showing that there is no reverse implications in Lemma \ref{hier}. In this paper, we are interested in the $\prec_c$ order, and previous Lemma gives a first idea on how this order compares to the majorization order. But we will only work with $\prec_c$ in this paper. Also note that the set of measures which are worst than a given $p \in \m(a)$ with respect to $\prec_c$
is given by the convex hull of the orbit of $p$ through cycles, whereas it is given by the convex hull of the orbit of $p$ through permutations when considering $\prec_d$. Note that if $p \in \sm(a)$, these two sets are the same.

\begin{definition}
For $p \in \m(a)$, we define the Fourier transform of $p$ by
\begin{align}
\ft(p)(\omega)= \sum_{k=0}^{a-1} p(k) e^{-2 \pi i k \omega/a}, \quad \omega \in \F_a
\end{align}
and the inverse Fourier transform of $h: \F_a \to \mC$ by
\begin{align}
\ft^{-1}(h)(k)= \frac{1}{a} \sum_{w=0}^{a-1} h(w) e^{2 \pi i k \omega/a}, \quad k \in \F_a.
\end{align}
\end{definition}

\begin{remark}\text{}\\
1. $\ft(p \star q) = \ft(p) \ft(q)$ for any $p,q \in \m(a)$. \\
2. If $p \in \sm_k(a)$ with $p(k)=1-P$, we have that $\ft(p)$ is given by $\ft(p)(0)=1$ and 
\begin{align}
& \ft(p) (\omega) = (1-\frac{aP}{a-1})e^{-2\pi i k \omega/a}, \quad \omega \neq 0.
\end{align}
\noindent
3. From previous remark, note that $(\sm(a), \star)$ is a semi-group.
\end{remark}

\begin{definition}
For $p \in \m(a)$, let $\mathrm{DOM}_c(p)$ be the set of probability measures which dominate $p$ with respect to $\prec_c$, i.e., $\mathrm{DOM}_c(p)= \{q \in \m(a): p \prec_c q\}$. 
\end{definition}

\begin{remark}\label{r2}
Note that it is easier to describe the set of measures that are dominated by a fixed measure $p$ than the reverse. However, we can write $\mathrm{DOM}_c(p)= \{q \in \m(a): \ft^{-1}(\ft(p)/\ft(q)) \geq 0\}$, and we can use the FFT algorithm to compute $\mathrm{DOM}_c(p)$ efficiently. 
\end{remark}
\begin{lemma}
For any $a \geq 1$, 
\begin{align}
p_1,p_2 \in \sm(a), \quad p_1 \prec_h p_2 \quad \Rightarrow \quad p_1 \prec_c p_2.
\end{align}
\end{lemma}
\begin{proof}
Assume that $p_1 \in \sm_k(a)$ and $p_2 \in \sm_l(a)$ for $k,l \in \F_a$, and $p_1 \prec_h p_2$. 
Then, denoting $1-P_1=p_1(k)$ and $1-P_2=p_2(l)$,
\begin{align}
\ft(p_1)(\omega) / \ft(p_2)(\omega) = \frac{1-\frac{aP_1}{a-1}}{1-\frac{aP_2}{a-1}} e^{-2\pi i (k \ominus_a l)/a}. \label{ratio}
\end{align}
Hence, if $(1-\frac{aP_1}{a-1})/(1-\frac{aP_2}{a-1}) \in  \mathrm{Im} (f)$,
where $f: P \in [0,1] \mapsto (1-aP/(a-1))$, we have that \eqref{ratio} is the Fourier transform of an element in $\sm(a)$. 
This is easily verified since $\mathrm{Im} (f)=[0, 1]$ 
and since by assumption $1-P_1 \leq 1-P_2$.
\end{proof}
Since $\sm(2)=\m(2)$, we have the following corollary. 
\begin{corol}\label{order2}
\begin{align}
p_1,p_2 \in \m(2), \quad p_1 \prec_h p_2 \quad \Rightarrow \quad p_1 \prec_c p_2.
\end{align}
\end{corol}

We now introduce one more ordering notion.
\begin{definition}
We define for $p_1,p_2 \in M(a)$
\begin{align*}
p_1 \prec_{cp} p_2 \,\,\, \text{iff} \,\,\, & p_1=p_2 \star \nu  \,\,\, \text{where $\nu$ is an infinitely divisible}\\
& \text{probability distribution.}
\end{align*}
\end{definition}
\begin{definition}
A probability distribution $p\in \m(a)$ is infinitely divisible if for any $k\geq 1$, there exists
$p_k \in \m(a)$ such that $$\displaystyle p=\star_{i=1}^k p_k,$$ 
or equivalently, if $\ft(\ft(p)^{1/k}) \geq 0$. 
\end{definition}
Note that checking the infinitely divisibility condition for a large enough $k$ implies the result for smaller $k$'s (by grouping the $p_k$'s). Hence, denoting $\e=1/k$, we need to check that $\ft(p)^\e$ has a valid inverse Fourier transform when $\e$ tends to 0. Let $z=\ft(p)$ and denote the component of $z$ by $z_j=r_j e^{i \theta_j}$. Then 
\begin{align*}
z_j^\e&= r_j^\e e^{i \e \theta_j}= (1+\e \log_e r_j)(1+ i \e \theta_j) + o(\e) \\ &= 1+\e (\log_e r_j + i  \theta_j) +o(\e).
\end{align*}
Hence, by the linearity of $\ft^{-1}$,
$$ \ft^{-1} (z^\e) = (1,0,\dots,0) + \e \ft^{-1}( (\log_e r_j + i  \theta_j)_{j=0}^{a-1}) + o(\e)$$
and to ensure $\ft^{-1} (z^\e) \geq 0$ for any $\e>0$, we need to ensure that
\begin{align}
& y(1),\dots,y(a-1) \geq 0 \label{lacond} \\
& \text{where} \,\,\, y=\ft^{-1} ((\log_e r_j + i  \theta_j)_{j=0}^{a-1}). \notag
\end{align}
Note that the dependency in $k$ has been removed in previous condition. 

To summarize: we have defined a notion of `convolution ordering', with $\prec_c$, where one can reach a distribution from another one with a circular convolution, and a notion of `convolutional path ordering', with $\prec_{cp}$, where one can reach the second distribution with small convolutional steps.
\section{Universality in polarization}\label{laref}
As mentioned in the introduction, there are two parts which require the knowledge of the source distribution in the source polar coding scheme: one in the compression and one in the reconstruction part. We present in this section two lemmas to be used for each of these parts in universal results. We start with the compression part.

\begin{definition}[Polar storage sets]\label{sset}
Let $\delta \in (0,1)$, $n$ a power of 2 and $p \in \m(a)$,
$$\cS_{\delta,n} (p) := \{ i \in [n] : H(U_i | U^{i-1}) \geq \delta \}$$
where $U^n =  X^n G_n$, $G_n=    \bigl[\begin{smallmatrix} 
      1 & 0 \\
      1 & 1 \\
   \end{smallmatrix}\bigr]^{\otimes \log_2(n)}$, $X^n \iid p$.
We use the notation
\begin{align*}
\cS(p_1) \supseteq  \cS(p_2) \,\,\, \text{if} \,\,\, \cS_{\delta,n}(p_1) \supseteq  \cS_{\delta,n}(p_2) \,\,\, \forall \delta \in (0,1), n. 
\end{align*}
\end{definition}

(We will sometimes call the components of $U^n$ on $\cS$ the information bits.)
The reason why we are interested in nested storage sets is clear: if one stores the components of a source distributed under $p_1$, it will also store the information components of any source $p_2$ with $\cS(p_1) \supseteq  \cS(p_2)$ (it will consume more rate than required for compressing a source under $p_2$ specifically, but it will allow lossless compression for both). However, for the reconstruction, it is not clear whereas the nested structure is sufficient to induce a universal decoding process. But let us postpone for now the reconstruction problem and focus on the nested structure only. 

\begin{lemma}\label{sensing}
For any $a \geq 1$, 
\begin{align}
p_1 \prec_c p_2 \quad \Rightarrow \quad \cS(p_2) \subseteq \cS(p_1).
\end{align}
\end{lemma}
\begin{proof}
By assumption, there exists $c \in \m(a)$ such that $p_1 = p_2 \star c$. 
Let $X^n \iid p_2$, $Z^n \iid c$ independent of $X^n$ and $\widetilde{X}^n=X^n \oplus Z^n \iid p_1$. 
Define $U^n=G_n X^n$, $\widetilde{U}^n=G_n \widetilde{X}^n$ and $W^n=G_n Z^n$, hence $ \widetilde{U}^n = U^n \oplus W^n$. We have
\begin{align}
H(\widetilde{U}_i | \widetilde{U}^{i-1}) &\geq H(\widetilde{U}_i | \widetilde{U}^{i-1}, W^n) \\
&= H(U_i | U^{i-1}, W^n)\\
&= H(U_i | U^{i-1}) \label{allo}
\end{align}
where the last equality follows from the fact that $U^n$ is independent of $W^n$ since $X^n$ is independent of $Z^n$. 
\end{proof}
Note that the ordering in \eqref{allo} indeed holds for all indices. 
We now investigate the reconstruction problem. We first recall the decoding algorithm used in \cite{ari}

\begin{definition}\label{polar-dec}[\texttt{polar-dec} algorithm \cite{ari,ari3}]\text{}\\
Inputs: $p \in \m(a)$, $n \in \mZ_+$, $\cS \subseteq [n]$ and $u[\cS] \in \F_a^{|\cS|}$.\\
Output: $\texttt{polar-dec}(p,u[\cS],n) \in \F_a^n$.\\
The algorithm proceeds as follows:\\
(0) Initialize $\mathcal{M}=\cS$;\\
(1) Find the smallest integer $i$ in $\mathcal{M}^c$ and compute \\
$u_i=\arg \max_{x \in \F_a} \pp_p\{ U_i=x | u[\mathcal{M}]\}$;\\
(2) Update $\mathcal{M} = \mathcal{M}\cup \{i\}$ and go back to (1) until $\mathcal{M}=[n]$;\\ 
(3) Output $x^n=u^n G_n$ where $u^n=u[\mathcal{M}]$.\\

The term $\pp_p\{ U_i=x | u[\mathcal{M}]\}$ is the probability that $U_i=x$ when $U[\mathcal{M}]$ is observed, where $U^n=X_nG_n$ and $X^n \iid p$. It is shown in \cite{ari,ari3} that the computational cost for each of these probabilities, as well as the overall algorithm, is bounded as $O(n \log_2 n)$ (more precisely $O(a^2 n \log_2 n)$ for the dependence in $a$ and $O(a \log_2(a) \cdot n \log_2 n)$ if \cite{mmac} is used and $a$ is power of 2). We refer to \cite{ari,ari3} for the recursive procedure to compute these probabilities, which uses a ``divide and conquer'' procedure based on the Kronecker structure of $G_n$.  
\end{definition}
\begin{definition}
Let $p_1,p_2 \in \m(a)$, $\delta \in (0,1)$, $n\geq 1$, $X^n \iid p_2$, and $\hat{X}^n=\texttt{polar-dec}(p_1,U[\cS_{\delta,n}(p_1)],n)$, where $U^n=X^n G_n$.
We define 
$$P_e(p_1|p_2) = \pp\{X^n \neq \hat{X}^n\} .$$
\end{definition}

\begin{lemma}\label{cplemma}
For any $a\geq 1$, $\delta \in (0,1/2)$, $n \geq 1$,
\begin{align*}
p_1 \prec_{cp} p_2 \,\,\, \Rightarrow \,\,\,  & P_e(p_1|p_2) \leq P_e(p_2|p_2).
\end{align*}
\end{lemma}
\begin{proof}
Fix $n$ and $\delta < 1/2$. If $p_1$ is the uniform distribution, $S_{\delta,n}(p_1)=\{1,\dots,n\}$ and the claim is clear: since we store all components, the left-hand side error probability is 0. 
Hence, assume that $p_1$ is not the uniform distribution.

Let us assume that $a=2$, the proof for $a>2$ is similar. 
Let $p_2 \in \m(2)$ and
$q = p_2 \star  1_\rho$, where $1_\rho= [1-\rho,\rho]$.
Since $q \prec_c p_2$, we have $\cS_{\delta,n} (q) \supseteq \cS_{\delta,n} (p_2)$ and for the components $i$ to be decoded
$$\delta > H(U_i | U^{i-1}) \geq  H(V_i | V^{i-1}),$$
where the $U_i$'s (resp. $V_i$'s) are i.i.d. under $q$ (resp.\ $p_2$).
For $W_1,\dots,W_n$ i.i.d. under $p$ and $w,w_1,\dots,w_{i-1} \in \{0,1\}$, define the mapping
\begin{align}
f_{w|w^{i-1}}: p \mapsto P(W_i = w| W^{i-1}= w^{i-1}). \label{cont}
\end{align}
Note that $f_{w|w^{i-1}}$ is continuous over $\m(2)$ (with the topology induced by $\mR^2$) for any $w,w_1,\dots,w_{i-1} \in \{0,1\}$. 

Also note that $H(U_i | U^{i-1}) < \delta$ implies that there exist $\xi(\delta)$ with $\xi(\delta) \stackrel{\delta \to 0}{\to} 0$, such that for any $u^{i-1} \in \{0,1\}^n$, 
\begin{align}
P(U_i = 0 | U^{i-1}= u^{i-1}) \wedge  P(U_i = 1 | U^{i-1}= u^{i-1}) < \xi(\e). \label{min}
\end{align}
Hence, using \eqref{min} and the continuity of $f_{w|w^{i-1}}$, we have for $\rho$ small enough and any $i$ in the complement of $\cS_{\delta,n} (q)$,  
\begin{align}
& \arg\max_{u \in \{0,1\}} P(U_i = u| U^{i-1}= u^{i-1}) \\ &= \arg\max_{u \in \{0,1\}} P(V_i = u| V^{i-1}= u^{i-1}) . \label{game}
\end{align}

For $p_1 \prec_{cp} p_2$, we have that  $p_1 = p_2 \star_{i=1}^k  1_\delta $ for any $k$ where $\delta$ depends on $k$ and is decreasing with $k$ increasing. We now want to iterate previous argument, but we have to use the continuity of \eqref{cont} at different distribution $q$'s, namely $q= p_2 \star_{i=1}^l  1_\delta$ for $l=1,\dots,k$. Since this is a compact set of $q$'s (the entire path from $p_2$ to $p_1$), we can pick $k$ large enough such that the continuity argument remains effective along the entire path, and \eqref{game} is proved by \eqref{min}. It is important to assume that $p_1$ is not uniform, so as to keep $\rho$ bounded below from 0. 

Finally, from \eqref{game}, we have that the algorithm \texttt{polar-dec} used with the mismatched distribution still leads to the same
output than when used with the matched one. Since in the mismatched scenario we observe more components than needed, strictly speaking we have an inequality in the error probability as in the lemma's statement.  

For $a \geq 3$, the proof is identical, except that we are moving along $p_2 \star_{i=1}^k \nu$ where $\nu$ is close to a delta function over $\F_a$, and \eqref{min}, resp. \eqref{game}, holds when the minimum, resp. maximum, include all elements of $\F_a$.
\end{proof}

This result tells us that, if we do the compression and the reconstruction using the distribution $p_1$, we can compress and reconstruct losslessly any source distribution which are better than $p_1$ in terms of $\prec_{cp}$.

If one were to use a compression scheme ignoring any complexity considerations, then, simply by knowing that the source distribution has an entropy at most $R$, it would be possible to compress and reconstruct the source losslessly, using the method of types for example, at rate $R$. And the set $\{p\in \m(a): h(p) \leq R \}$ are essentially the largest sets which can be compressed losslessly at a fixed rate.
It is ambitious to ask for such a ``broad universality'' with polar codes, since these are structured codes with complexity attributes, in contrast to the codes derived with the method of types. We may have to give up some extra rate to achieve this goal, or we may universally compress only certain subsets of source distributions. We now investigate these points.  

\subsection{Results for binary sources}\label{main1}
For binary sources, it is possible to achieve a broad universal result with polar coding. 

Notation: For a given $0\leq R<1$, let $p_0(R), p_1(R)$ be the two binary probability distributions such that $H(p_0(R))=H(p_1(R))=R$. 

\begin{definition}[Universal polar compression algorithm]\text{}

A. Compression:\\
Inputs: $R \in [0,1]$ (the rate of compression), $\delta$ (the target error probability), $x \in \F_2^n$ (the data).\\
Output: $v \in \F_2^{nR+o(n)}$ (the stored data).\\
The compression algorithm proceeds as follows:\\ 
1. Compute $u= G_n x$\\
2. Store $u[\cS_{\delta,n}(p_0(R))]$ and $u[n]$

B. Reconstruction:\\
Inputs: $n$, $R$, $u[\cS_{\delta,n}(p_0(R))]$ and $u[n]$\\
Outputs: $\texttt{polar-dec-adapt}(p_0(R),p_1(R),u[\cS_{\delta,n}(p_0(R))],$\\$u[n],n)$

\end{definition}

\begin{definition}[\texttt{polar-dec-adapt} algorithm]\text{}\\
Inputs: $p_1,\ldots,p_k \in \m(a)$, $n \in \mZ_+$, $\cS \subseteq [n]$, $\cT \subseteq \cS^c$ (called the set of checkers) and $u[\cT \cup \cS] \in \F_a^{|\cT|+|\cS|}$.\\
Output: $\texttt{polar-dec-adapt}(p_1,\dots,p_k, u[\cS], u[\cT],n) \in \F_a^n$.\\
The algorithm proceeds as follows:\\
(1) For $j=1,\dots,k$, run $u_{(j)}^n =\texttt{polar-dec}(p_j,u[\cS],n)$\\
(2) Find $t=\arg\min_{j=1,\dots, k} d_H(u_{(j)}[\cT] , u[\cT])$ (pick one at random for ties)\\
(3) Output $\texttt{polar-dec}(p_t,u[\cS],n)$. (Variant: output $\texttt{polar-dec}(p_t,u[\cS\cup \cT],n)$.)
\end{definition}

\begin{thm}\label{mainthm}[Universal polar compression]
Let $X^n=[X_1,\dots,X_n]$ be i.i.d.\ Bernoulli with $H(X_1) \leq R$.
The universal polar compression algorithm allows to compress $X^n$ at rate R, with error probability $O(2^{-n^{\beta}})$, for any $\beta<1/2$, and compression/reconstruction complexity $O(n \log_2 n)$.
\end{thm}
(Note: Using the duality argument of Section \ref{duality}, this theorem admits an analogue for universal coding over binary symmetric channels.)

For the proof of this theorem, we show that:\\
1. Any binary source which is known to have entropy at most $R$ can be compressed universally with polar codes by storing the information bits on $\cS(p^*)$, where $p^*$ is one of the two distributions with entropy $R$.\\
2. If it is known on which symbol the source distribution puts more mass, the source can also be losslessly reconstructed with \verb+polar-dec+ using a checker.\\
3. If it is not known on which symbol the source distribution puts more mass, the source can also be losslessly reconstructed using the modified decoding algorithm \verb+polar-dec-adapt+.


\begin{proof}[Proof of Theorem \ref{mainthm}]
Let $D(R) \subseteq \m(2)$ be the set of binary distributions with entropy at most $R$, and as before, denote by $p_0(R)$ and $p_1(R)$ the 
two distributions of entropy equal to $R$ (assume $R <1$, the result is otherwise trivial).  
Note that, by Corollary \ref{order2}, for $i=0,1$  
$$p_i(R ) \prec_c D(R), $$
and by Lemma \ref{sensing}
$$\cS_{\delta,n}(p_i(R)) \supseteq  \cS_{\delta,n}(p),\quad \forall p 
\in D(r), \delta,n.$$
Hence, by storing the components on $\cS_{\delta,n}(p_0(R))$, we are not loosing any information bits. 
We have to set $\delta = \delta_n = 2^{-n^{\alpha}}$ with $\alpha<1/2$ large enough to reach the desired $\beta$ in the Theorem.

Let $D_i(R)$, $i=0,1$, be the two regions of $D(R)$ containing distributions putting more mass on 0, respectively 1, assuming consistent indexing with $p_0(R)$ and $p_1(R)$. 
Note that for $i=0,1$
$$p_i(R) \prec_{cp} D_i(R).$$
Hence, if we know that the source distribution belongs to $D_0(R)$, 
we can conclude from Lemma \ref{cplemma} that $\texttt{polar-dec}(p^*,u[\cS_{\delta,n}(p^*)],n)$ leads to an exact recovery,
with error probability at most equal to the error probability of the source polar scheme designed with perfect knowledge of the source distribution, which is from \cite{ari3}, $O(2^{-n^{\beta}})$, for any $\beta<1/2$. From the same paper, we conclude that the compression and reconstruction complexity $O(n \log_2 n)$.

If we do not know whether the true distribution is in $D_0(R)$ or $D_1(R)$, we can learn it as follows. 
Assume that the type of $X^n$ is close to its Bernoulli distribution; this is not the case with an exponentially small probability.  
Notice that the observed data $U[\cS_{\delta,n}(p_0(R))]$ corresponds to an equally likely string under both a distribution in $D_0(R)$ and $D_1(R)$, since
the distribution on $\cS_{\delta,n}(p_0(R))$ when $X^n$ is drawn under $p_0(R)$ or $p_1(R)$ is uniform.  
Say w.l.o.g.\ that the true distribution, $p_*$, is in $D_0(R)$. If we use $p_0(R)$ for $\texttt{polar-dec}$, we will get the right output (modulo the error probability).
If we use $p_1(R)$, we will recover $\tilde{X}^n$ which is typical under $\tilde{p}_*$, the measure obtained by $p_*$ by exchanging the probability mass at 0 and 1.
To see this, note that for a typical $X^n$ under $p_*$, $X^n + 1^n$ is typical under $\tilde{p}_*$ (where $1^n$ is the $n$-dimensional vector filled with 1's). Moreover, $$ 1^n G_n = [0^{n-1}, 1],$$ and the last component of $U^n$ cannot be in $\cS_{\delta,n}(p)$ (the last component of $U^n$ is the one with least conditional entropy) unless $p$ is the uniform distribution. Hence $X^n$ and $X^n + 1^n$ are both typical upon observing $U[\cS_{\delta,n}(p_0(R))]$, and we must have been able to recover correctly $X^n$ or $X^n + 1^n$ when knowing if the true distribution was in $D_0(R)$ or $D_1(R)$ and decoding with respectively $p_0(R)$ or $p_1(R)$. 
Hence, by storing the value of $U[n]$ (even if it has low entropy) and running the algorithm $\texttt{polar-dec-adapt}$ with both $p_0(R)$ and $p_1(R)$, we can check which one of the two models provides the correct estimate for $U[n]$ and learn whether $p_*$ is in $D_0(R)$ or $D_1(R)$. Indeed, there is no need to run twice the algorithm, it is sufficient to run it once and use the value of $U[n]$ which had been stored. In any case, we will make en error, if $\texttt{polar-dec}$ fails, which happens from \cite{ari3} with probability $O(2^{-n^{\beta}})$, for any $\beta < 1/2$, and the complexity of this scheme is $O(n \log_2 n)$. 
\end{proof}

\subsection{Results for $a$-ary sources}\label{main2}
\begin{definition}
For $D \subset \m(a)$, let
\begin{align}
 p_c(D)& := \arg \min_{p \in \m(a): p \prec_c D} H(p), \label{domp}\\
 \hat{p}_c(D)&  := \arg \min_{p\in \sm(a): p \prec_c D} H(p). \label{domp2}
\end{align}
In any of the above minimization, if the minimizer is not unique, pick one arbitrarily. 
\end{definition}
We clearly have that $\hat{p}_c(D) \prec_h p_c(D)$; however, it is trivial to find $\hat{p}_c(D)$ while finding $p_c(D)$ is more difficult. 

Let us assume for instance that the exact source distribution is unknown for the compression part, but is known for the reconstruction part. If the source distribution is known to belong to a set $D \subset \m(a)$, one way to construct the storage set is to pick $\cS(p_c(D))$. Then, from Lemma \ref{sensing} and Corollary \ref{order2}, this retains the information bits of any source in $D$. Of course, this may consume more rate than needed with an optimal source code, in other words, if we define $H_{\text{max}} (D) : = \max_{p \in D} H(p),$
we have in general
\begin{align}
H(p_c(D)) \geq H_{\text{max}} (D). \label{ineq}
\end{align}
The inequality can be strict since $p_1 \prec_h p_2$ does not imply in general $p_1 \prec_c p_2$, and there are examples where equality holds in the above, in which case a compression designed for $p_c(D)$ requires the minimal rate to compress any source in $D$. 
\begin{remark}\label{domino}
Let $D \subset \m(a)$ be such that $\arg \max_{p \in D} H(p)$ is unique (denoted $p_h(D)$) and satisfies $p_h(D) \prec_c D$. Then a source polar code designed for $p_h(D)$ can compress any source in $D$ at the lowest achievable rate $\max_{p \in D} H(p)$ without loosing any information bits. 
\end{remark}
(Note that $p_h(D) \prec_c D$ implies that $p_c(D)=p_h(D)$.) 
The set $\mathrm{DOM}_c(p)$, plotted in Figure \ref{shape}, satisfies (by definition) the condition of Remark \ref{domino} for any $p$. Comparing Figure \ref{shape} with the plot of $\mathrm{DOM}_h(p):=\{q \in m(a): q \prec_h p\}$ also shows that there are sets for which \eqref{ineq} holds with a strict inequality. 
One can also check how much rate is lost by compressing for a distribution that is dominated in terms of $\prec_c$ as opposed to $\prec_h$, for example 
the gap between the rate needed to compress $B_R:=\{p\in \m(a): H(p) \leq R \}$ using $p_c(B_R)$ and the minimal rate $R$ needed with the method of types, i.e., $H(p_c (B_R)) - R$.
This gap can be computed using Remark \ref{r2}, in the case of Figure \ref{shape} for example, it is 0.095 for R=0.865 and $a=3$.
\begin{figure}
\begin{center}
\includegraphics[scale=.3]{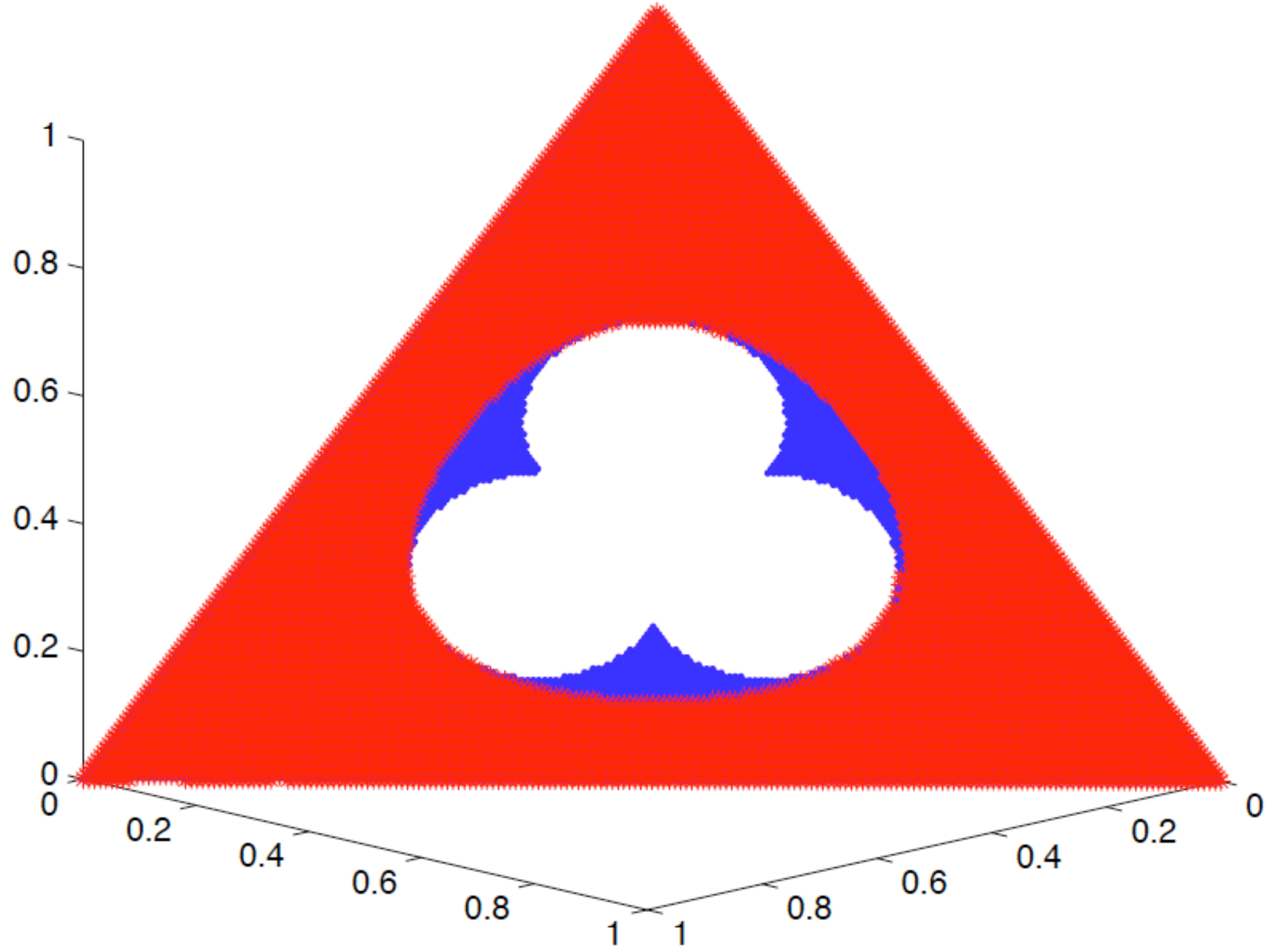}
\caption{Plots of $\mathrm{DOM}_h([0.2,0.2,0.6])$ (red region) included in $\mathrm{DOM}_c([0.2,0.4,0.4])$ (blue region).}
\label{shape}
\end{center}
\end{figure}

Also note that $p_1 \nprec_c p_2$ may not imply $\cS(p_2) \nsubseteq \cS(p_1)$, and there may be other ways to construct storage sets which contain the information bits for several distributions (than using $\prec_c$). We investigate this point in Section \ref{imp} and now move to the decoding part for $a$-ary sources.

\begin{definition}
For $D \subset \m(a)$, let
\begin{align}
 p_{cp}(D) := \arg \min_{p \in \m(a): p \prec_{cp} D} H(p). \label{pcp}
\end{align}
If the minimizer is not unique, pick one arbitrarily. 
\end{definition}

The following follows by definitions. 
\begin{lemma}
A source distribution known to belong to a set $D \subset \m(a)$ can be compressed and reconstructed losslessly at rate $H(p_{cp}(D))$, using a polar code designed for the distribution $p_{cp}(D)$.
\end{lemma}


\subsection{Non-universality of $a$-ary source polar codes}\label{neg}

In this section, we show that in general, polar codes cannot achieve the lowest rate for lossless compression of compound sources when $a\geq 3$, no matter how the storages sets are constructed (i.e., not necessarily via $\prec_c$).
A similar result has been derived in \cite{rudi1} for channel coding, however, it is not possible to leverage the counter-example found in \cite{rudi1} to the source case (since the channel polarization results have a source counter-part only for additive noise channels, and the counter-example in \cite{rudi1} does not use only with additive noise channels). In this section, we assume that for the decoding part, we have the aid of a genie that provides the exact source distribution. 

We consider two source distributions $p$ and $q$ on $\F_a$, and we are interested in finding the rates at which one can compress these two sources without loosing the information bits of any of them. 
We denote by $C_{\textrm{pol}}(p,q)$ the infimum of these rates, and we provide different bounds on this quantity.
Clearly
$$C(p,q):=H(p) \vee H(q) \leq C_{\textrm{pol}}(p,q).$$ 
From previous section, we have the upper bound $C_{\textrm{pol}}(p,q) \leq  H(p_c(p,q)),$
where $p_c(p,q)$ is as defined in \eqref{domp} for the set $D=\{p,q\}$.
In our definition, $C_{\textrm{pol}}(p,q)$ is given by the limit inferior of 
$$ \frac{1}{n} | \cS_\delta (p) \cup \cS_\delta (q) |.$$ 
Let $n=2^\ell$, $U^n=G_nX^n$ where $X^n$ is i.i.d. under $p$, and $V^n=G_nY^n$ where $Y^n$ is i.i.d. under $q$. Let us also denote by $P$ (resp. $Q$) the additive noise channel whose noise distribution is $p$ (resp. $q$). We then have from Section \ref{duality}
\begin{align}
H(U_i|U^{i-1}) = 1-I(P_i), \quad H(V_i|V^{i-1}) = 1-I(Q_i)
\end{align}
where $P_i$ (resp. $Q_i$) are the channels corresponding to $P^{\sigma}$ for $\sigma \in \{-,+\}^\ell$, as defined in \cite{ari} with the tree construction. 
 Moreover, if we define for $\delta \in (0,1)$
$\mathcal{G}_\delta (P) = \{ i \in \{1,\dots,n\} : I(P_i) \geq \delta \}$, 
we have
\begin{align}
\cS_\delta (p) \cup \cS_\delta (q) = \left( \mathcal{G}_\delta (P) \cap \mathcal{G}_\delta (Q) \right)^c. \label{corresp}
\end{align}
This shows that the compound capacity for source or channel coding are related and we can use the result of Section \ref{duality} and Theorem 5 in \cite{rudi1} to get the following bounds.
\begin{lemma}\label{rud}
\begin{align}
C_{\textrm{pol}}(p,q)  &\leq \frac{1}{2^\ell}  \sum_{\sigma \in \{-,+\}^\ell}  I(\text{BEC}(Z(P^\sigma)) \vee I(Z(Q^\sigma)))\\
C_{\textrm{pol}}(p,q) & \geq \frac{1}{2^\ell}  \sum_{\sigma \in \{-,+\}^\ell}  H(p^\sigma) \vee H(q^\sigma)   
\end{align} 
where $P$ (resp. $Q$) is the additive noise channel with noise distribution $p$ (resp. $q$).
Moreover each bound is monotonically approaching $C_{\textrm{pol}}(p,q)$.
\end{lemma}
Note that the upper bound is straightforward, and the notation $H(p^\sigma)$ refers to $H(U_i|U^{i-1})$ for the index $i$ corresponding to $\sigma$. It is interesting to note that if BECs can be used to compute previous bounds, we cannot use the counter-example of \cite{rudi1} to show that polar codes do not achieve compound capacity in source coding, since BECs do not correspond to a valid source distribution via the duality of Section \ref{duality}.
However, we can use the duality and BECs to construct storage sets which are included in $\cS_\delta (p) \cup \cS_\delta (q)$, in a different manner than done in previous section. Let us give an example with $\ell=1$. For two source distributions $p$ and $q$, consider finding the BECs with parameter $Z(P)$ and $Z(Q)$ ($P$ and $Q$ as defined above). Then, as in \cite{rudi1}, the good indices for $P$ and $Q$ satisfy 
\begin{align}
& \mathcal{G} (P) \cap \mathcal{G} (Q) \supset \mathcal{G} (\text{BEC}(Z(P))) \cap \mathcal{G} (\text{BEC}(Z(Q))) \\
& \equiv \mathcal{G} (\text{BEC}(Z(P)\vee Z(Q)))
\end{align}
and from \eqref{corresp}, $\mathcal{G} (\text{BEC}(Z(P)\vee Z(Q)))$ gives a storage set to compress $p$ and $q$ without loosing information bits. This provides an interesting and different approach to constructing universal polar codes, although it may not be practical and has the drawback of requiring the source distribution for the reconstruction (as opposed to the $\prec_{cp}$ ordering). In a work in progress, we propose the use of spike measures $\sm(a)$ to replace the ``worst BECs" directly with ``worst source distributions''. The common feature between the spike measures and BECs is that they are both families that have a nested structures for the storage/good index sets and that span the whole range of entropy/mutual information between 0 and 1. Also note that as opposed to the channel polarization case, degradedness in source polarization is less restrictive, since there are less degrees of freedom for source distributions than channels.

Now, to show that polar codes do not achieve the compound capacity in source coding, we can still use the lower bound of Lemma \ref{rud}, but we need to pick two source distributions on ternary source alphabets. 
\begin{prop}
Polar codes do not achieve the compound capacity for source coding when the source alphabet has strictly more than 2 elements. 
\end{prop}
{\it Counter-example:} Let $p=[0.08,0.36,0.56]$, $q=[0.11,0.62,0.27]$, such that $H(p)=0.8143$, $H(q)=0.8126$ and $C=H(p) \vee H(q)=0.8143$. The LHS of Lemma \ref{rud} for $\ell=1$ evaluates at $0.8174$ which is strictly larger than $C$.



\section{Sketching and sparse recovery}\label{cs}
In compressed sensing (CS), a $k$-sparse signal of high dimensionality $n$ can be recovered with overwhelming probability from a small number of random measurements $m=O(k \log(n/k))$ with a convex optimization method \cite{candes,donoho}. If the use of random measurement matrices simplifies the mathematical analysis, a drawback is that they have a heavy structure and it there is no efficient way to check if a given matrix realization satisfies the desired property (RIP) for the reconstruction (although one can show that this happens with high probability). Other drawbacks of random sensing matrices are discussed in \cite{piotr, calderbank}. It has hence become a challenging problem to construct explicit matrices that, yet, can perform competitively (in terms of measurement rate) with the random ones. Different deterministic matrices have been proposed in the literature, but in \cite{cormode,muth,devore07} the number of rows is at least quadratic in $k$ and in \cite{devore06} one needs $\Omega(n)$ bits to specify a matrix entry. In \cite{piotr}, binary matrices with $m$ brought down to $k 2^{O(\log \log n)^E}=kn^{o(1)}$, with $E>1$, are proposed and in \cite{calderbank}, a rather general condition for constructing deterministic matrices satisfying a statistical restricted isometric property (STRIP) is given.  

In this section, we are interested in designing an explicit measurement matrix using the polarization technique. The motivation being that the matrix used in previous section for polar source compression is deterministic and easily constructed.  Of course, in the basic source compression problem, the compression matrix is designed adaptively to the source distribution, whereas the CS results are universal as long as the signal is sparse. Hence, we would like to construct an explicit matrix with the polarization technique that is also universal. The tools of previous section for universal polarization will hence be used. 

 
Note that there are a few more distinctions between CS and the problems of previous sections. First, the source in our case is a random process, whereas in the original works on compressed sensing, the signal is deterministic. The case of random signals has been considered in several subsequent works for compressed sensing, such as in \cite{calderbank}. Another important difference, is that the source in our setting is valued in $\F_a$, as opposed to $\mR$ for arbitrary sparse signals. One way to address this problem is via quantization, which requires a careful treatment. 
It is related to the fact that in the CS setting, measurements can be done with arbitrary precision whereas in our setting they are quantified in bits.  On the other hand, we focus here on applications where the signal is valued in a discrete set to start with, such as in certain network monitoring problems \cite{varghese, gilbert}. For example, if one wishes to track the number of packets flowing between the different IP addresses of a network (e.g. to detect unusual behaviors), the state vector can be of dimension up to $2^{32}$. Since it is not feasible to maintain such a huge dimensional vector, one wishes to use a much smaller sketch vector that is still carrying all the significant information of the state vector, by exploiting the fact that the state vector is sparse. 
We hence keep such applications as our motivation and focus mainly on the sketching (sensing) and sparse recovery of such discrete signals. A possible lifting of the results in this paper to the real field setting is investigated in a work in progress.  

Before attacking the problem of constructing a universal deterministic sketching matrix via the polarization technique, we consider a specific example by choosing a particular distribution for the signal to get started. 

\subsection{Assuming knowledge of the signal distribution}\label{toogood}
In this section, we assume that the distribution of the signal is known.
The case of unknown distributions is discussed in Section \ref{univprior}. 
Assume that $X_1,\dots,X_n$ are i.i.d.\ under the following spike distribution
\begin{align}
p_{\e} := (1-\e,\e/(a-1),\dots,\e/(a-1)) \in \sm_0(a). \label{labelle}
\end{align}
Note that for $n$ i.i.d.\ samples drawn under $p_\e$, the number of non-zero components is in expectation $n\e$. 
\begin{definition}[Polar sketching matrix for a single distribution]\label{phi1}
Let $\delta \in (0,1)$ and $\phi_{\delta,n}(p_\e) = I_{\cS} \cdot G_n$ be the matrix obtained by deleting the rows of $G_n$ which are not indexed by $\cS=\cS_{\delta,n}(p_\e)$ (cf.\ Definition \ref{sset}).
\end{definition} 
Rephrasing the source polarization result, we obtain the following. 
\begin{lemma}\label{know}
Let $n$ be a power of 2, $X$ be an $n$-dimensional vector drawn i.i.d. under $p_\e$, and let $\phi=\phi_{\delta_n, n}(p_\e)$ be the polar sketching matrix defined for $p_\e$ and $\delta_n=2^{-n^{\beta}}$ with $\beta \in (0,1/2)$ (cf.\ Definition \ref{phi1}). For any $\alpha \in (0,1/2)$, there exists $\beta \in (0,1/2)$ such that the number of rows of $\phi$ is given by
$$m= n (1-\e)\log_a(\frac{1}{1-\e})+n \e \log_a(\frac{a-1}{\e}) + O(2^{-n^{\alpha}})$$
and using the polar decoding algorithm for $p_\e$, we can recover $X$ from $Y=\phi X$ with probability $O(2^{-n^{\alpha}})$ and with a complexity bonded as $O(a^2 n \log_2 n)$ (and if $a$ is a power of 2, the complexity can be reduced to $O(a \log_2 a \cdot n \log_2 n)$ by using the approach in \cite{mmac}.). 
 \end{lemma}

{\it Discussion:} Note that $m$ is simply the cardinality of $\cS$, which is approximately $n H(p_\e)$.
Defining $n\e=k$, we have 
\begin{align}
m&=  k \log_{a} \frac{n}{k} + o(\log_{a} \frac{n}{k}) \label{m1} \\
&=\frac{1}{\log_e a} k \log_{e} \frac{n}{k} + o(\log_{e} \frac{n}{k}) .
\end{align}
This expression is similar to the $O(k \log_e \frac{n}{k})$ expression encountered in the CS literature (\cite{candes}) for the number of measurements. It is even a tighter form since the constant is less than 2; of course, for the reasons discussed at the beginning of Section \ref{cs}, the comparison is inappropriate, since (in particular) we are modifying the assumption on the signal: it is drawn from a specific known distribution. The reason why the number of measurements decreases when $a$ increases may seem strange; however notice that a measurement for signals in $\F_a$ is made with a precision of $a$ bits. Hence, to compare the number of measurements for different values of $a$, one should use the same unit for the measurement. Let us check how the number of measurements scale with $a$. Rewriting \eqref{m1} with the dependency in $a$, we have
\begin{align}
m&=  k (1+\log_{a} (a-1)) + k \log_{a} \frac{n}{k} 
& = 2 k + o_a(1).
\end{align}
Hence, if we allow infinite precision for the measurements, for large $a$ we only need $2k$ measurements, but of course, the complexity blows up. 
If we express all measurements in nats, we have
\begin{align}
m&=  k (1+\log_{e} (a-1)) + k \log_{e} \frac{n}{k} 
\end{align}
and $m$, as a function of $a$, grows like $k \log_e a$.

Of course, Lemma \ref{know} requires knowledge of the signal distribution whereas CS results are universal. There is no reason to assume that \eqref{labelle} is distribution of the signal. The exact knowledge of the source distribution is in general unrealistic, and as discussed in previous section, even with an estimate of the distribution, it is crucial to show at least some robustness with respect to possible mismatched distributions. For applications, it may actually be interesting to have adaptive results, but this is also changing the rules of the game. We now investigate the universality problem.

%


\subsection{Universal prior}\label{univprior}
A possible way of defining $k$-sparse random sources, is to ask that the source distribution leads to an expected number of at most $k$ non-zero values. 
 Specifically, let $a$ be a prime number and let $\F_a=\{0,\dots,a-1\}$. Let 
$X^n=(X_1,\dots,X_n)$ be i.i.d. samples from a distribution $\mu$, with $\mu(0)=1-\e$.
Then, the number $K(X^n)$ of components of $X^n$ which are not equal to $0$ is in expectation
\begin{align}
\E K(X^n) = n\e.
\end{align}
Let
\begin{align}
\spa(a,\e) :=\{ \mu \in \m(a) : \mu(0) \geq 1-\e\},
\end{align}
and consider samples $X^n=(X_1,\dots,X_n)$ that are i.i.d. from a distribution in $\spa(a,\e)$. From previous remark, the number of components in $X^n$ that are not equal to $0$ is bounded by $n\e$. If we are interested in $\e$ as a measure of sparsity, then $K(X^n)/n$ concentrates exponentially fast around $\e$. The set $\spa(a,\e)$ is pictured in Figure \ref{spafig} for $a=3$.
\begin{figure}
\begin{center}
\includegraphics[scale=.3]{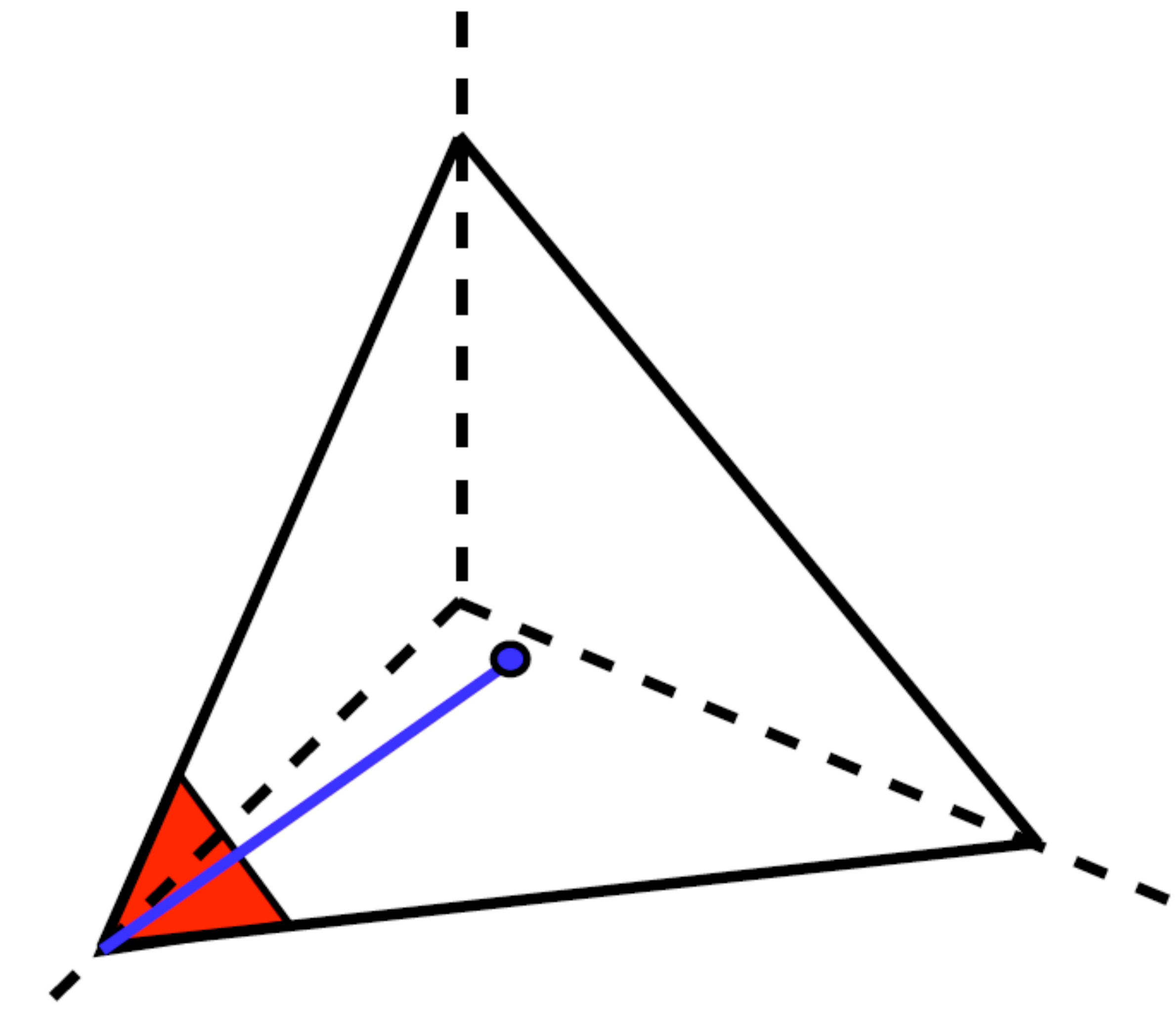}
\caption{The simplex with $\spa(3,\e)$ (lower triangle, in red) for $\e \approx 1/5$ and the spike measures at 0, namely $\sm_0(3)$ (middle line, in blue).}
\label{spafig}
\end{center}
\end{figure}

The results that we will derive do not depend on the fact that 0 is the special value for the distributions in $\spa(a,\e)$, in other words, we could equally well consider sources that are sparse with respect to an arbitrary $i \in \F_a$. 
For simplicity, we stick with $i=0$ for now, although considering arbitrary $i$'s may be useful when dealing with the problem of quantizing a signal to $\F_a$.
Also note that $\spa(a,\e)$ contains sources which can be supported on any subset of $\F_a$ (e.g., $a$ may be large but this set still contains sparse binary sources). It may be reasonable to assume that there is no such variation in the probability mass assigned to the non-special values, this will be discussed later.

\begin{remark}\label{shannon}
From an information-theoretic point of view, we can ask the question of finding the smallest rate at which one could compress a source whose distribution is in $\spa(a,\e)$ without any further knowledge on the distribution (and irrespectively of the compression scheme employed). As discussed in Section \ref{laref}, the answer to this question is given by the maximal entropy that can be reached with a distribution in $\spa(a,\e)$. It turns out that the distribution with maximal entropy is precisely \eqref{labelle}, as in Section \ref{toogood}. Hence, theoretically, the strong performances presented in Section \ref{toogood} can still hold in the universal setting. However, the scheme used to achieve such Shannon limiting performance may be highly complex, whereas the whole point here, is to consider explicit  schemes of low complexity. 
\end{remark}


\begin{thm}\label{csprop}
Let $X^n$, with $n$ a power of 2, be an $n$-sample drawn i.i.d. from a distribution which has at most $\e$ mass on the non-zero entries of $\F_a$, and 
let $\phi(a,\e)$ be the $m \times n$ polar sketching matrix constructed deterministically for $a$ and $\e$ (cf.\ Definition \ref{smatrix}). We have 
\begin{align*}
&m = C(a,\e) \cdot k \log_e \frac{n}{k} +O(k), \quad k= n \e,
\intertext{with}
&\lim_{\e \to 0}C(a,\e)=\frac{a-1}{\log_e a}
\end{align*} 
and with probability $1-O(2^{-n^{\beta}})$, $\beta \in (0,1/2)$, $X^n$ can be exactly reconstructed from $\phi X^n$ using the polar decoding algorithm 
(cf.\ Remark \ref{r4}) with a complexity of $O(a^2 n \log_2 n)$ (or $O(a \log_2 a \cdot n \log_2 n)$ if $a$ is a power of 2 and \cite{mmac} is used). 
\end{thm}
\begin{remark}\label{r4}\text{}\\
1. The polar decoding algorithm (cf.\ Definition \ref{polar-dec}) must be evaluated as 
$$\texttt{polar-dec}(p_{cp}(\spa(a,\e)),\phi x^n[\cS_{\delta,n}(p_{cp}(\spa(a,\e)))],n)$$ for $\delta= \delta_n = 2^{-n^{\alpha}}$ with $\alpha<1/2$ large enough to reach the desired $\beta$ in the Theorem, and 
where $p_{cp}(\spa(a,\e))$ is the distribution of minimal entropy that is dominated by the entire set $\spa(a,\e)$ for $\prec_{cp}$ as in \eqref{pcp}, and $\cS$ as in Definition \eqref{sset}. \\
2. The multiplication $\phi X^n$ is carried out over $\F_a$.\\
3. The same result holds if the distribution of $X^n$ has at most $\e$ mass on an arbitrary $i \in \F_a$.\\
4. In Section \ref{imp}, we discuss improvements of the constant $C$.
\end{remark}


\begin{definition}\label{smatrix}
Given a set $D$ of probability measures on $\F_a$, we construct a sketching matrix $\phi(D)$ of dimension $n$ as follows:
\begin{itemize}
\item[(i)] Find $p_{cp}(D)$ as defined in \eqref{pcp}
\item[(ii)] Find $\cS=\cS_{\delta, n} (p_c(D))$ as in Definition \ref{sset} for $0<\delta< 1$
\item[(iii)] Define $\phi =  I_\cS G_n  ,$
where $G_n=    \bigl[\begin{smallmatrix} 
      1 & 1 \\
      0 & 1 \\
   \end{smallmatrix}\bigr]^{\otimes \log_2 n}$ and where $I_\cS$ is the matrix whose columns indexed by $\cS$ form the identity matrix and whose other columns are filled in with zeros. Note that $\phi$ is an $m \times n$ matrix, where $m=|\cS|$.
\end{itemize}
In particular, we define $\phi(a,\e):=\phi (\spa(a,\e))$ and to have the optimal error decay we pick $\delta = 2^{-n^{\alpha}}$ with $\alpha<1/2$.
\end{definition}

{\it Implementation of $\phi$.} \\
1. Step (i) can be easily computed, cf.\ Remark \ref{r2} and the proof below.\\
2. Step (ii) requires a comment: finding $\cS$ with an analytic formula is a hard open problem in polar codes. However, it is mostly a mathematical challenge, since one can run simulations to determine $\cS$ with a good accuracy, or find arbitrarily tight bounds on the entropy terms in $\cS$ in polynomial time \cite{vardy2}.\\
3. The construction of $G_n$ is straightforward because of its Kronecker structure, which also allows an efficient decoding algorithm running in $O(n \log_2 n)$. 

\subsection{Interpretation of Theorem \ref{csprop}}
In view of 
$$\lim_{\e \to 0} C(a,\e)= \frac{a-1}{\log_e(a)},$$
for a fixed small $\e$, the quantization level $a$ should be at most $1/\e$ in order to have a dimensionality reduction. 
Hence, if a signal sparser in its domain than its magnitude, 
where we define the magnitude-sparsity of a signal taking $a$ possible values by $1/(a-1)$, and the domain-sparsity as before by $\e=k/n$, then the approach of Theorem \ref{csprop} gives interesting results.

For $a$ small, this is interesting for most $\e$.
In particular for $a=2$, the sparsity in magnitude is maximal, namely 1, and for any $\e$ we have an optimal dimensionality reduction
$$m=1.44 \cdot k \log_e (n/k),$$
where the optimality refers here not only to the order $k \log_e (n/k)$ but also to the constant $1.44$ (recall that the measurements are taken in bits). By Shannon, one cannot further improve this bound (even with schemes of high complexity). 

For $a$ large, e.g. $a=257$, we get reasonable dimensionality reduction for very sparse data, for example, if 
$\e= 10^{-3}$ and $n=10^6$, we get a reduction of $68\%$ for the number of measurements (compared to $n$). But for $a=257$ and $\e=0.1$, there is almost no dimensionality reduction. However, we will see in next section that this is due to the analysis employed in the proof of Theorem \ref{csprop} rather than the use of the polar matrix.

\subsection{Proof of Theorem \ref{csprop}}
\begin{proof}
The number of measurements $m$ is given by $nH(p_{cp}(\spa(a,\e)))+o(n)$.
Note that by symmetry, $p_{cp}(\spa(a,\e))$ is a spike measure (i.e., an element of $\sm_0(a)$). 
We have $$p_{cp}(\spa(a,\e))=(1-\eta(\e), \eta(\e)/(q-1), \dots, \eta(\e)/(q-1))$$ where
$\eta(\e)$ is the smallest positive $\eta$ ensuring 
$$ (1-\eta , \eta /(q-1), \dots, \eta /(q-1)) \prec_{cp} p$$
for any $p \in \spa(a,\e)$.
Moreover, it is sufficient to check 
$$ (1-\eta , \eta  /(q-1), \dots, \eta  /(q-1)) \prec_{cp} (1-\e,\e,0,\dots,0) ,$$
i.e.,
$$\ft^{-1} \left(\frac{\ft(1-\eta, \eta/(q-1), \dots, \eta/(q-1))}{\ft(1-\e,\e,0,\dots,0)} \right)^{1/k}  \geq 0$$
for any $k \geq 1$.
Using \eqref{lacond}, the dependence in $k$ can also be removed.
Defining
\begin{align}
z&= \frac{\ft(1-\eta, \eta/(q-1), \dots, \eta/(q-1))}{\ft(1-\e,\e,0,\dots,0)} \\
&=\frac{(1,1-\eta \frac{q}{q-1},\dots,1-\eta \frac{q}{q-1})}{(1-\e + \e e^{-2 \pi i t/a})_{t=0}^{a-1}},
\end{align}
and denoting the component of $z$ by $z_j=r_j e^{i \theta_j}$, we need to ensure 
\begin{align}
& y(1),\dots,y(a-1) \geq 0 \label{lacond} \\
& \text{where} \,\,\, y=\ft^{-1} ((\log_e r_j + i  \theta_j)_{j=0}^{a-1}). \notag
\end{align}
Numerically, one can then easily find $\eta(\e)$ by means of the FFT algorithm. 
In Figure \ref{delta_e}, we have plotted $\e \mapsto \eta(\e)$ for different values of $a$.
Note that one can also find analytically $\eta(\e)$ using the following approach. 
Assume $a=3$.
Let us first find $p_{c}(\spa(3,\e))$. Here also, we have $p_{c}(\spa(3,\e))=(1-\bar{\eta}(\e),\bar{\eta}(\e)/2,\bar{\eta}(\e)/2 )$ where we need to find $\bar{\eta}(\e)$.
Note that all distributions that are worst than $(1-\e,\e,0)$ for $\prec_c$ are given by the convex hull of the
orbit of $(1-\e,\e,0)$ though cycles, that is $\mathrm{hull}((1-\e,\e,0),(0,1-\e,\e),(\e,0,1-\e) )$.  
Hence, the projection $p_{cp}$ of $(1-\e,\e,0)$, i.e., the distribution in this convex hull which belongs to the spike measures and has minimal entropy is found by taking the intersection between the line connecting $(1-\e,\e,0)$ to $(\e,0,1-\e)$ and the line of spike measures parametrized by $(1-d, d/2,d/2)$. An elementary computation yields
\begin{align}
\bar{\eta}(\e)=1-2 \e (1-\e) = 1-2 \e +o(\e).
\end{align}
Note that the scaling $1-2\e +o(\e)$ is clear, since for small $\e$ the line connecting $(1-\e,\e,0)$ to $(\e,0,1-\e)$ is almost parallel to the line connecting $(1,0,0)$ to $(0,0,1)$.
Indeed, one can easily generalizes this for $a>3$ to
\begin{align}
\bar{\eta}_a(\e)=  1-(a-1) \e +o(\e).
\end{align}

To find the projection $p_{cp}$ of $(1-\e,\e,0)$, we need to move from $(1-\e,\e,0)$ towards spike measures with tiny convolutional steps.
But once we have made a small step in the direction $(1-\e,\e,0) - (\e,0,1-\e)$ to reach $(x,y,z)$, we need to move next in the rotated picture, i.e., in the direction $(x,y,z)-(y,z,x)$, as illustrated in Figure \ref{rotatefig}.
\begin{figure}
\begin{center}
\includegraphics[scale=.3]{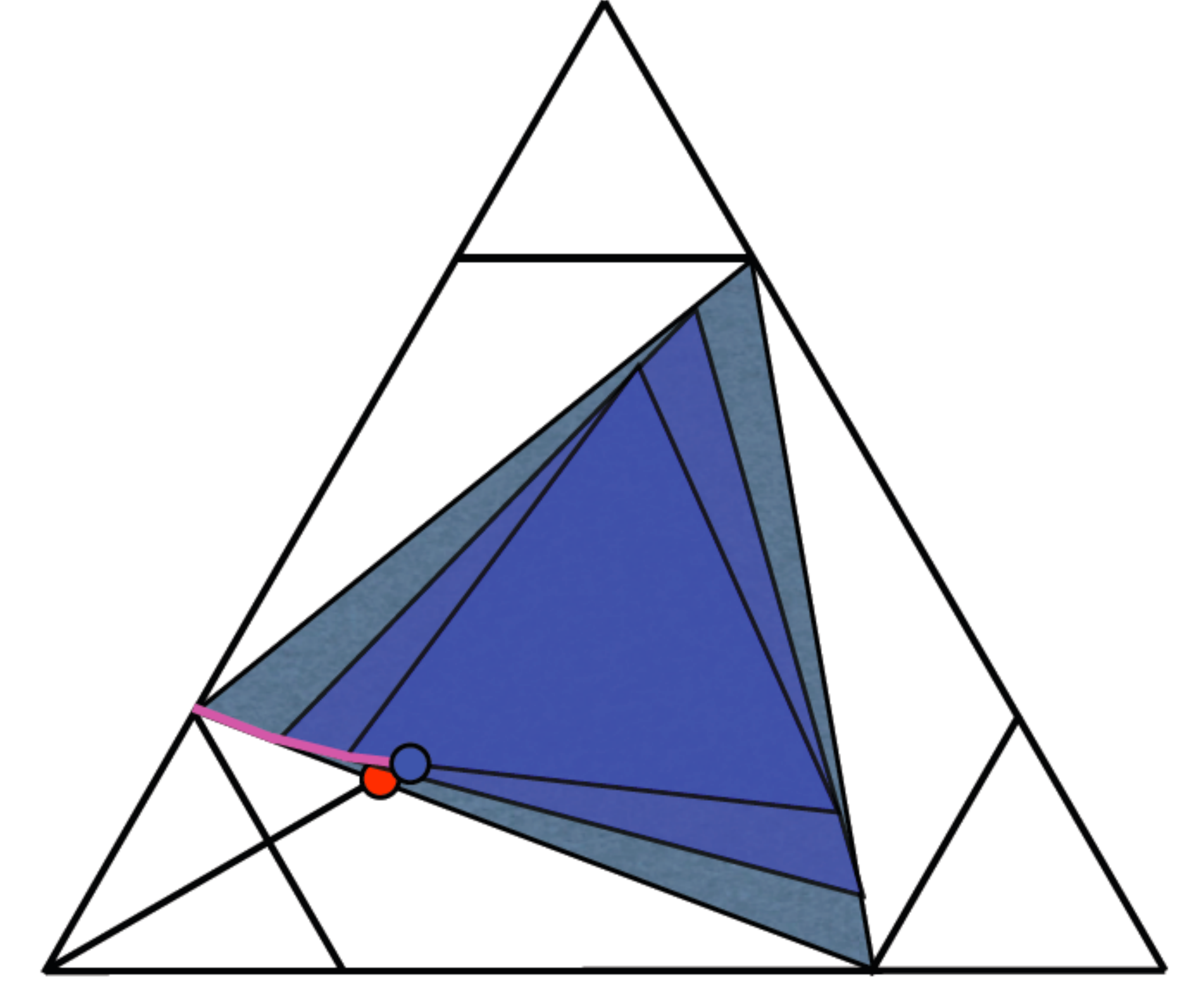}
\caption{The simplex with the $p_{c}$ projection (first point in red) and the $p_{cp}$ projection (second point in blue) of $[1-\e,0,\e]$.}
\label{spafig}
\end{center}
\end{figure}
Defining $f(x) = x + \gamma (\Pi - I )x$, where $\Pi=\circ(0,1,0)$, we are interested in $f^k (x)$ where $x=(1-\e,\e,0)$.
Hence, we look for $A^k$ where $A=I + \gamma (\Pi - I)$. Since $\Pi$ is circulant, so is $A$ and the eigenvector of $A$ are the Fourier (DFT) basis elements and the eigenvalues are $1+\gamma (\lambda_i -1)$, where $\lambda_i$ are the corresponding 3 roots of unity. Therefore, the eigenvalues of $A^k$ are $[1+\gamma (\lambda_i -1)]^k$, and keeping $\gamma k=c$, we obtain $$\lim_{k \to \infty} [1+c/k (\lambda_i -1)]^k = \exp(c(\lambda_i-1)).$$
Hence, $$\tau(c)=F_3 \diag (\exp(c(\lambda_i-1))) F_3^* (1-\e,\e,0)^t,$$ for $c\geq 0$, and where $F_3$ is the Fourier (DFT) matrix of dimension 3, parametrizes the path starting at $(1-\e,\e,0)$ and obtained with incremental convolutional steps which are ``targeting'' spike measures. Equating the second and first components of $\tau_c$, i.e., solving $\tau(c)_2=\tau(c)_3$ gives a closed form expression for $c$ and for $\eta(\e)=2 \tau(c)_2$, and we get as for $\bar{\eta}(\e)$,
\begin{align}
\eta (\e)= 2 \e +o(\e).
\end{align}
That is, for small $\e$, the penalty endured by considering the $p_{cp}$ projection rather than the $p_c$ one (for $\spa$) is negligible. This is not surprising, since for $\e$ small, the path from $(1-\e,\e,0)$ to spike measures is anyway small (as required for the $p_{cp}$ projection).
With similar arguments, we conclude that for any $a$, 
\begin{align}
\eta_a (\e)= (a-1) \e +o(\e).
\end{align}
Finally, we need to evaluate $H(p_{\eta(\e)})$
where $$p_{\eta(\e)}=(1-\eta(\e), \eta(\e)/(a-1),\dots,\eta(\e)/(a-1)).$$
Note we can compare the cost for universality of a low complexity scheme (obtained with the $p_{cp}$ analysis and the polar matrix) with respect to the limiting performance (Shannon): instead of $nH(p_{\e})$ 
we need $nH(p_{(a-1)\e})$ measurements, when $\e$ is small. For $a=2$, these two are identical, and this is consistent with Theorem \ref{mainthm}. For arbitrary $a$, we get
\begin{align}
H(p_{\eta(\e)}) = \frac{a-1}{\log_e a} \e  \log_e \frac{1}{\e} + O(\e )
\end{align}
as opposed to 
\begin{align}
H(p_{\e}) = \frac{1}{\log_e a} \e  \log_e \frac{1}{\e} + O(\e ).
\end{align}

\begin{figure}
\begin{center}
\includegraphics[scale=.365]{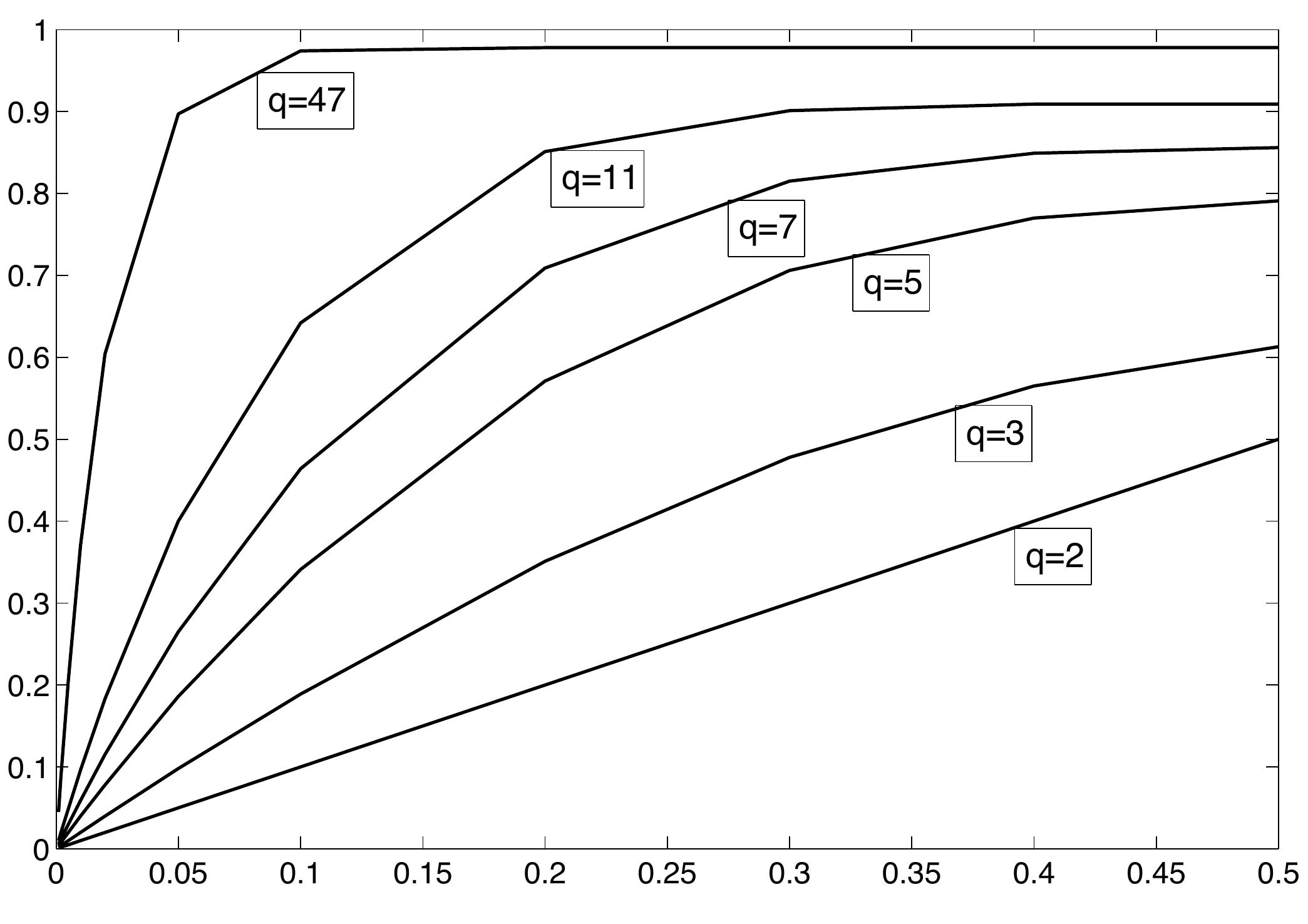}
\caption{Plots of $\e \mapsto \eta(\e)$ for different values of $a=q$. Note that $(a-1)\e$ is an upper bound to $\eta(\e)$, tight for small $\e$ as shown in the proof of Theorem \ref{csprop}. Hence for $\e$ not too small $\eta(\e)$ provides a better measurement rate than what is obtained with the crude bound of Theorem \ref{csprop}.}
\label{delta_e}f
\end{center}
\end{figure}

\begin{figure}
\begin{center}
\includegraphics[scale=.365]{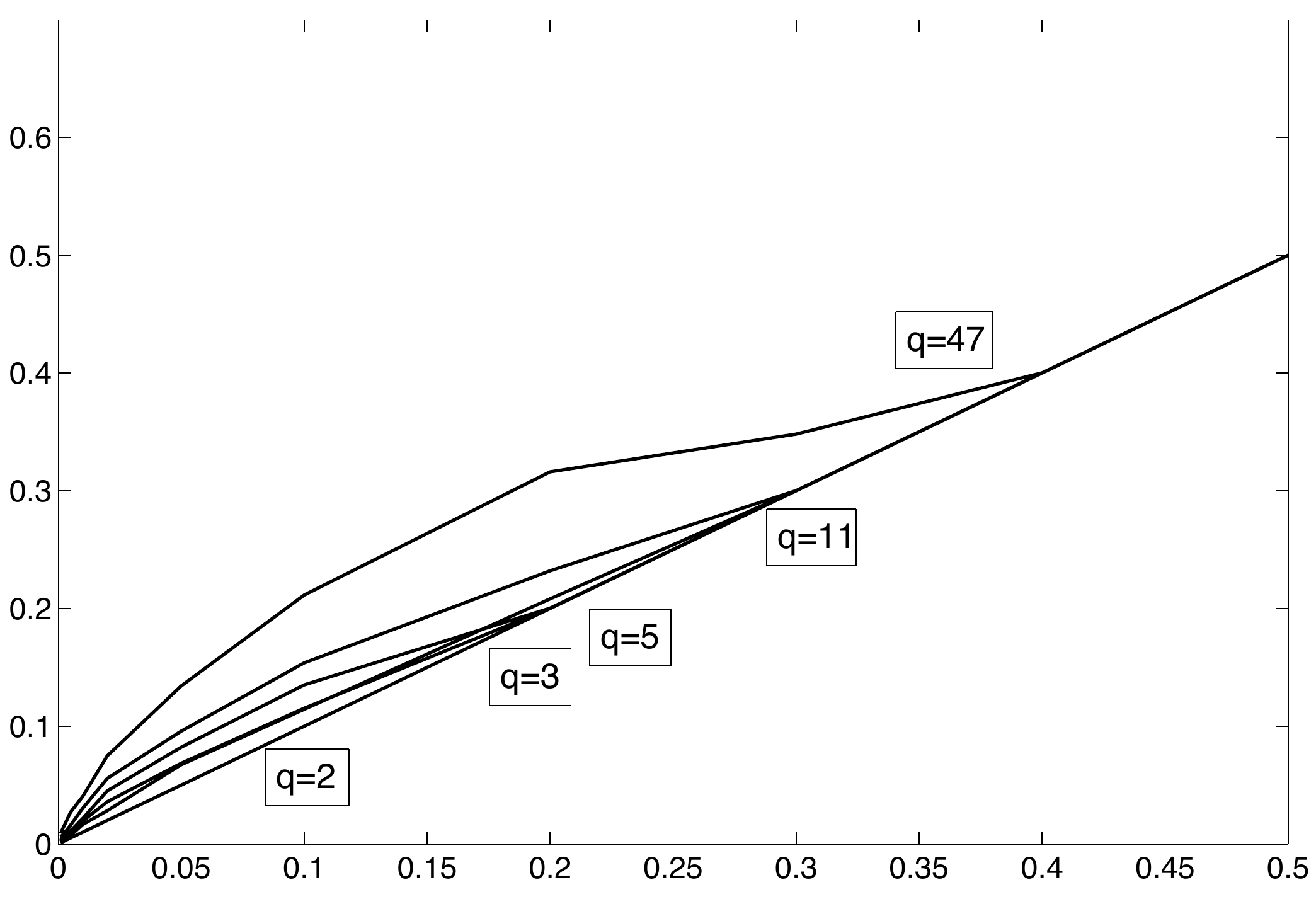}
\caption{Plots of $\e \mapsto \eta^*(\e)$ for different values of $a=q$}
\label{delta_estar}
\end{center}
\end{figure}

\end{proof}

\subsection{Improving Theorem \ref{csprop}}\label{imp}

In Remark \ref{shannon}, we concluded that the minimal number of measurements (in $a$-ary bits) needed to recover a source from $\spa(a,\e)$ is given by $n H(p_\e)$ where $p_\e$ is the spike measure with mass $1-\e$ at 0 and $\e/(a-1)$ elsewhere. To approach this performance with a polarization scheme of low complexity is challenging. We have shown that, for the universal problem of recovering sources from any distribution in $\spa(a,\e)$, one could still use the adaptive setting but designed for a specific distribution; namely $p_{cp}(\spa(a,\e))$. This distribution is dominated by the entire set $\spa(a,\e)$ with respect to $\prec_{cp}$, and we showed that because of this property, a scheme designed for $p_{cp}(\spa(a,\e))$ guarantees successful recovery for any distributions in $\spa(a,\e)$. In a sense, $p_{cp}(\spa(a,\e))$ is the worst case scenario. Of course, there may be other ways (than using $\prec_{cp}$) to order distributions and find a worst case distribution which has lower entropy (one may attempt to replace $\prec_c$ with $\prec_d$). There may even be other ways of tackling the universal problem than looking for an ordering and a worst case distribution. The advantage of using a worst case approach, is that we can then inherit the complexity attributes and convergence rate property from the adaptive setting. Ordering also allows to give a `hierarchy' between the different distributions, and helps designing robust schemes (by backing off from the estimated distribution to guarantee performances).

One point is that it may not be needed to consider the entire set $\spa(a,\e)$ for a given problem. For example, the set of distributions that are dominated by $p_{cp}(\spa(a,\e))$ (w.r.\ to $\prec_{cp}$), is already quite large and contains most distributions of sparsity $\e$. It does not contain the distributions of sparsity $\e$ which have small supports, but if these can be ruled out, then we can bring back the constant $C(a,\e)$ to $1/\log_{e}(a)$, which is the Shannon limiting performance.

We now discuss another approach to construct a universal sketching and reconstruction method for $\spa(a,\e)$. A before, there are two parts to discuss. First the sketching, i.e., to know which rows of $G_n$ can be deleted without loosing information about $X^n$. Then the reconstruction, i.e., to know how to run a decoding algorithm that ignores the exact distribution.

\subsubsection{Sketching}
Here is a brut-force approach to achieve universal sketching.
\begin{definition}[\texttt{brut-univ-sketching} algorithm]\text{}\\
Inputs: $\e$ (the sparsity degree), $a$ (the size of the signal alphabet), $n$ (the dimension).\\
Outputs: $\eta^*(\e)$.

We present two variants of the algorithm. \\ 
Variant A:\\
For $\eta$ from $\e$ until $\eta(\e)$ (with a given step size);\\
if $\cS(p_{\eta}) \supseteq \cS(q)$ for any $q$ in the convex hull of 
$(1-\e,\e,0,\dots,0)$, $(1-\e,0,\e,0,\dots,0)$,...,$(1-\e,0,\dots,0,\e)$;\\
$\text{}\quad$ output $\eta$;\\
 otherwise increase the step size.\\
Variant B:\\
For $\eta$ from $\e$ until $\eta(\e)$ (with a given step size);\\ 
if $\cS(p_{\eta}) \supseteq \cS(q^{(\e)})$ where $q^{(\e)}=(1-\e,\e,0,\dots,0)$;\\
$\text{}\quad$ output $\eta$;\\ otherwise increase the step size.
\end{definition}

Note: one could also consider a dichotomic approach for the search of $\eta^*(\e)$,
and by symmetry, one can restrict the search of $q$ to only one portion of the convex hull.
One also has to specify the precision $\xi$ for the computations of the sets $\cS=\cS_{\xi,n}$, we 
omitted it in the algorithm to simplify the notation.

Variant B has low complexity, since it conducts a search in a one-dimensional space (for $\eta$) and since the computation of $\cS$ can be done at low computational costs. Variant A requires a larger search for $q$, which can be constraining for $a$ large. 

{\it Result:} Variant A of \texttt{brut-univ-sketching} provides $\eta^*(\e)$ such that $$\cS(p_{\eta^*(\e)}) \supseteq \cS(p), \quad \forall p \in \spa(a,\e).$$
{\it Conjecture:}  Variant B of \texttt{brut-univ-sketching} leads to the same output than Variant A. 

In Figure \ref{delta_estar}, we show $\eta^*(\e)$ (obtained with Variant B of \texttt{brut-univ-sketching}) and the Shannon limit consisting of the diagonal, and reached for $a=2$. As observed, the improvement is significant with respect to $\eta(\e)$ (obtained with the $p_{cp}$ projection). Indeed, this brings the number of measurement very close to the optimal performance. Emre Telatar is gratefully acknowledged for his help in producing these plots. 

For the decoding part, there are no guarantees that decoding with $p_{\eta^*(\e)}$ allows a correct recovery. One can use the algorithm \texttt{polar-dec-adapt} to learn the distribution, but one needs to first add checkers in the set of stored components. Checkers are components that need not to be stored (because they have low conditional entropy) but that we still store to help the decoder get information about the source distributions. As long as the number of checkers is $o(n)$, the measurement rate is not affected. 

\subsubsection{Reconstruction} 
We now proceed to use \texttt{polar-dec-adapt} to decoder the sensed components of previous part. We first proceed to a patching of $\spa(a,\e)$.
Consider a uniform discretization of the convex hull of
$(1-\e,\e,0,\dots,0)$, $(1-\e,0,\e,0,\dots,0)$,...,$(1-\e,0,\dots,0,\e)$. 
Enumerate a uniform discretization of this convex hull as $D_k:=\{p_1,\dots,p_d\}$.  
Call $\spa_d(a,\e)$ the sets of distributions that dominates any of the elements in $D_k$ with respect to $\prec_{cp}$. We then have
$$\spa_d(a,\e) \to \spa(a,\e),$$
meaning that the set $\spa_d(a,\e)$ is dense in $\spa (a,\e)$. 
For a targeted $\e$, we then pick a $\e'$ slightly larger than $\e$ and a $d$ large enough such that
$$\spa_d(a,\e') \supseteq \spa(a,\e).$$
We then use \texttt{polar-dec-adapt} with the output of \texttt{brut-univ-sketching} and $o(n)$ checkers to learn which of the distributions $p_k$ is a good `model' for the sensed data. The term model is used because by construction of $\spa_d(a,\e')$, the sensed data on the checkers must look typical with at least one of the $p_i$'s, although there might be more than one, and although none of these $p_i$'s may be the true distribution of $X^n$ (but they will be dominated with respect to $\prec_{cp}$ by the true distribution, which is good enough to ensure correct decoding). 
One has to pick $\e'-\e$ small enough and $d$ large enough to ensure a small increase in the number of measurements (one needs to study the scaling of these parameters for a more precise statement). The overall complexity of this decoding scheme scales multiplicatively with $d$. Hence, as long as $d$ is not of the order of $n\log_2 n$, the overall complexity remains low. 

In a work in progress, we also consider another approach for universal decoding via an algebraic characterization of the possible likelihood ratios computed in \texttt{polar-dec}.

\section{Discussion and extensions}
\subsection{Universal polar coding}
We summarize here two ideas introduced in this paper to construct universal polar coding schemes
\begin{enumerate}
\item Convolutional path ordering: to tell when a polar coding scheme designed for one distribution can succeed for another one
\item Checkers: to learn some information about the distribution by storing components that did not need to be stored
\end{enumerate}

In particular, we developed an algorithm which allows to compress universally binary sources at the lowest achievable rate, with low complexity and with guaranteed low error probability.

\subsection{Sparse recovery and sketching}

We applied the tools developed for universal polar coding to the problem of sketching sparse signals,
constructing a deterministic sketching matrix by deleting appropriate rows of the polar matrix $G_n$.
We summarize here some conclusions and extensions on this approach. 

\begin{enumerate}
\item {\it An sketching method tuned to discrete signals.}\\
Compressed sensing exploits the sparsity of signals in their {\it domain} to acquire them efficiently. 
If for the application of interest, the signal is also sparse in its {\it magnitude}, that is, if it takes values in a set of small cardinality, this can also be exploited as shown in this paper. For example, if the signal is binary, we developed a sketching method with a deterministic low complexity matrix, an optimal number of measurements (for the scaling and the constant) and a low complexity recovery algorithm with a proved exponentially small (in $\sqrt{n}$) probability of error. We extended this results to $a$-ary vectors, noticing a better fit for small $a$ and proposing an improved approach for larger $a$ (Section \ref{imp}). 
We also underline that, for a given application, the method proposed here can be used adaptively by designing an appropriate probabilistic model for the signal. This can improve the measurement rate. 

\item {\it Lifting this work to the reals?}\\
Most works in the CS literature constructing explicit sensing matrices are based on algebraic constructions \cite{piotr,devore07,calderbank}. In these works, matrix acting on the reals can then be obtained. Of course, we also made the point (previous item) that for certain application, it may be more natural to work with the discrete setting directly. Yet, an interesting extension would be to study a lifting of our results to the real case. 
A possible approach would be via a quantization procedure, where problems of robustness to noise must be investigated. Another possible problem would be to attempt detecting the signal support only (which is a binary signal). 


\end{enumerate}

\end{document}